\documentclass[runningheads]{llncs}
\usepackage[utf8]{inputenc}

\usepackage{graphicx}
\usepackage{amsmath}
\usepackage{amssymb}
\usepackage{bbm}
\usepackage{physics}

\usepackage[hyphens]{url}
\usepackage{hyperref}
\usepackage[hyphenbreaks]{breakurl}

\usepackage[maxbibnames=20,style=numeric]{biblatex}
\usepackage{authblk}

\usepackage{tikz}
\usetikzlibrary{positioning}

\usepackage{algorithm}
\usepackage{algpseudocode}

\usepackage{varwidth}

\newcommand{\ZZ}{\mathbb{Z}}

\newcommand{\NN}{\mathbb{N}}

\newcommand{\crsdss}{{\tt CRS-DSS}}

\newcommand{\kg}{{\tt KeyGen}}
\newcommand{\peg}{{\tt P}}
\newcommand{\vic}{{\tt V}}

\newcommand{\getsr}{\stackrel{\$}{\gets}}

\newcommand\repeatedtheorem[2]{%
  \expandafter\let\expandafter\repeatedtheoremtmp\csname#1name\endcsname
  \expandafter\def\csname#1name\endcsname{%
    \repeatedtheoremtmp\ \ref{#2}%
    \global\expandafter\let\csname#1name\endcsname\repeatedtheoremtmp
  }
}

\spnewtheorem*{theorem*}{Theorem}{\normalshape\bfseries}{\itshape}

\spnewtheorem{pro}{Protocol}{\normalshape\bfseries}{\normalshape}

\newcommand{\modd}[1]{[#1]_n}

\newcommand{\bigO}[1]{\mathcal{O}(#1)}

\title{SPDH-Sign: towards Efficient, Post-quantum Group-based Signatures}
\titlerunning{SPDHSign}
\author{}
\institute{}

\author{
  Christopher~Battarbee\inst{1}
     \and
     Delaram~Kahrobaei\inst{1,2,3,4}
    \and
     Ludovic~Perret\inst{5}
     \and
     Siamak~F.~Shahandashti\inst{1}
 }

 \institute{
     Department of Computer Science, University of York, UK
     \and
     Departments of Computer Science and Mathematics, Queens College, City University of New York, USA
     \and
     Initiative for the Theoretical Sciences, Graduate Center, City University of New York, USA
     \and
     Department of Computer Science and Engineering, Tandon School of Engineering, New York University, USA
     \and
     Sorbonne University, CNRS, LIP6, PolSys, Paris, France
 }

\authorrunning{C. Battarbee, D. Kahrobaei, L. Perret and S. F. Shahandashti}

\begin{document}
\maketitle

\begin{abstract} 
In this paper, we present a new diverse class of post-quantum group-based Digital Signature Schemes (DSS). The approach is significantly different from previous examples of group-based digital signatures and adopts the framework of group action-based cryptography: we show that each finite group defines a group action relative to the semidirect product of the group by its automorphism group, and give security bounds on the resulting signature scheme in terms of the group-theoretic computational problem known as the Semidirect Discrete Logarithm Problem (SDLP). Crucially, we make progress towards being able to efficiently compute the novel group action, and give an example of a parameterised family of groups for which the group action can be computed for any parameters, thereby negating the need for expensive offline computation or inclusion of redundancy required in other schemes of this type.


\end{abstract}

\keywords{Group-based Signature \and Post-quantum Signature \and Group Action Based Cryptography \and Post-quantum Group-based Cryptography}

\section*{Introduction}
Since the advent of Shor's algorithm and related quantum cryptanalysis, it has been a major concern to search for quantum-resistant alternatives to traditional public-key cryptosystems. The resultant field of study is known today as Post-Quantum Cryptography (PQC), and has received significant attention since the announcement of the NIST standardisation. 

One of the goals of PQC is to develop a quantum-resistant Digital Signature Scheme (DSS), a widely applicable class of cryptographic scheme providing certain authenticity guarantees. Following multiple rounds of analysis, NIST have selected three such schemes for standardisation, two of which are based on the popular algebraic notion of a lattice. Nevertheless, stressing the importance of diversity amongst the post-quantum roster, a call for efficient DSS proposals not based on lattices was issued in 2022 \cite{nist4}. A potential source of post-quantum hard computational problems come from group-based cryptography; for a comprehensive survey of the field including examples of DSSs, see the work of Kahrobaei et al in \cite{kahrobaeinotices}, \cite{Kahrobaei-BattarbeeBook}.

Recall that a finite commutative group action consists of a finite abelian group $G$, a finite set $X$, and a function mapping pairs in $G\times X$ into $G$. Another promising framework for PQC has its origins in the so-called \textit{Hard Homogenous Spaces} of Couveignes\footnote{Similar notions were arrived at independently by Rostovstev and Stolbunov \cite{rostovtsev2006public}, \cite{stolbunov2010constructing}.} \cite{couveignes2006hard}: one considers a family of group actions for which all the `reasonable' operations - for example, evaluating the group action function, and sampling uniformly from the group - can be done efficiently, but a natural analogue of the discrete logarithm problem called the Vectorisation Problem is computationally difficult. Given such a group action, one can exploit the commutativity of the group operation to construct a generalisation of the Diffie-Hellman Key Exchange protocol based on the difficulty of the Vectorisation Problem, which is believed to be post-quantum hard.

As well as this analogue of Diffie-Hellman, the group action framework is used to construct an interactive proof of identity, which is effectively a standard three-pass identification scheme. In his doctoral thesis \cite{stolbunov2012cryptographic}, Stolbunov uses this identification scheme to obtain a signature scheme by applying the standard Fiat-Shamir heuristic; we will here follow the convention of referring to this scheme as the CRS\footnote{Couveignes, Rostovstev and Stolbunov.}  Digital Signature Scheme (\crsdss{}). In order to specify a practical signature scheme it remains to specify a group action: very roughly, \crsdss{} uses the celebrated example, coming from the theory of isogenous elliptic curves, of a finite abelian group called the `class group' acting on a set of elliptic curves.

\crsdss{} did not recieve much attention for a number of years, for two key reasons: first, it was demonstrated that the scheme admits an attack of quantum subexponential complexity \cite{childs2014constructing} (in fact, this attack applies to all group-action based cryptography). This might in itself be tolerable; much more troubling is that the original version of \crsdss{} is unacceptably slow. There has, however, been a resurgence of interest in schemes similar to \crsdss{} following the discovery in \cite{castryck2018csidh} of a much faster isogeny-based group action; on the other hand, the computation of the class group is in general thought to be computationally difficult. In fact this is quite a significant problem: without random sampling the security proofs, which rely on group elements hiding secrets to have the appropriate distribution, break down. Two approaches to solving this problem have been suggested: in \cite{de2019seasign}, one uses the `Fiat-Shamir with aborts' technique developed by Lyubashevsky \cite{lyubashevsky2009fiat}, at the cost of rendering the scheme considerably less space efficient; in \cite{beullens2019csi}, a state-of-the-art computation of a class group is performed and the resulting group action is used as the platform for \crsdss{}. However, it is important to note that here the computation of \textit{a} class group is performed, and so one is restricted in terms of tweaking parameters. In particular, the introduction of new parameters would require another extremely expensive offline class group computation. 

A potential third solution is to dispense with the isogeny-based group action altogether, and search for different examples of group actions for which computing the appropriate group - and therefore uniform sampling - is efficient. Historically speaking, there has not been much research in this direction since non-trivial examples of cryptographically interesting group actions have not been available - though this work is predated by a general framework for actions by semigroups in \cite{monico2002semirings}, and an example semigroup action arising from semirings in \cite{maze2005public}. In this paper we make an important step in the search for efficient group actions; in particular we show that every finite group gives rise to a group action on which \crsdss{}-type signatures can be constructed, and that the respective group is cyclic and has order dividing a known quantity. These group actions arise from the group-theoretic notion of the semidirect product, and were first studied in the context of a generalisation of Diffie-Hellman \cite{habeeb2013public} - note, however, that it was not known at the time that the proposed framework was an example of a group action. Indeed, the link was only discovered rather recently \cite{battarbee2022subexponential}, and prompted the isogeny-style renaming of the key exchange in \cite{habeeb2013public} as \textbf{S}emidirect \textbf{P}roduct \textbf{D}iffie-\textbf{H}ellman, or SPDH (to be pronounced `spud'). With this in mind, in this paper we propose a hypothetical family of digital signature schemes which we christen {\tt SPDH-Sign}.

It is important to note that we do not provide concrete security parameters, nor do we claim a security improvement over similar schemes: instead, the paper has two key contributions. First, we notify the community of a promising step towards efficient, scalable sampling in cryptographic group actions: our Theorem~\ref{thm:n-div} shows that for each group action we construct there is quite a severe restriction on the possible sizes of the cyclic group acting. Since sampling from a cyclic group is trivial if we know its order, we have provided a large class of candidate group actions for which sampling is efficient. As such we also carry out the standard methodology of defining a resulting signature scheme, and give a security proof in the random oracle model that bounds the security of the signature scheme in terms of our central algorithmic problem in more explicit terms than comparable proofs. 

The second key contribution is the proposal of a specific group as an example of a group in which one can efficiently sample in the resulting group action whilst maintaining resistance to related (but not known equivalent) cryptanalysis. Here we see an example of our Theorem~\ref{thm:n-div} in action - the size of the crucial parameter needed for efficient sampling can be one of only 12 values, and we can check the validity of each of these values in logarithmic time.

\subsection*{Related Work}
The following is a short note to emphasise the novelty of our contribution with respect to related areas of the literature. 

The idea of defining cryptography based on the action of a semigroup on a set, and the resulting “semigroup action problem” (SAP), is proposed in Chris Monico’s thesis \cite{monico2002semirings}, and is referenced by Han and Zhuang in their recent paper \cite{han2022dlp}. Certainly this idea of a semigroup action predates our establishment of a cryptographically relevant group action arising from topics in group theory. We therefore clarify that our contribution is not the novel proposal of a group action of this type, but the explicit connection between cryptographic group actions and the problems arising from semidirect product key exchange, which originally appears in \cite{habeeb2013public}.

In \cite{han2022dlp}, SAP and the semidirect product key exchange are mentioned in the same breath in the introduction. This, however, does not constitute the explicit connection of the problem originally appearing in the discussion of semidirect product key exchange and cryptographic group actions - where this connection is one of the claimed novel aspects of our paper - but a list of problems related to the semigroup DLP. Moreover, none of the semidirect product key exchange-adjacent literature we are aware of mentions SAP, including proposals of semidirect product key exchange \cite{habeeb2013public, kahrobaei2016using, rahman2022make, rahmanmobs} and cryptanalysis of the semidirect product key exchange authored by Monico himself \cite{monico2020remark, monico2021remarks}. Accordingly, we believe that establishing the connection between semidirect product key exchange and group-action based cryptography is a novel contribution to the area.


\section{Preliminaries}
\subsection{The Semidirect Product}
The term `semidirect' product refers, generally speaking, to a rather deep family of notions describing the structure of one group with respect to two other groups. For our purposes we are interested in a rather specific case of the semidirect product, defined as follows:

\begin{definition}
    Let $G$ be a finite group and $Aut(G)$ its automorphism group. Suppose that the set $G\times Aut(G)$ is endowed with the following operation:
    \[(g,\phi)(g',\phi')=(\phi'(g)g',\phi'\phi)\]
    where the multiplication is that of the underlying group $G$, and the automorphism $\phi'\phi$ is the automorphism obtained by first applying $\phi$, and then $\phi'$. We denote this group $G\ltimes Aut(G)$.
\end{definition}

A few facts about this construction are standard.
\begin{proposition}
Let $G$ be a finite group and $\Phi\leq Aut(G)$ (where $\Phi$ can be any subgroup, including $Aut(G)$ itself). One has the following: 
\begin{enumerate}
    \item $G\ltimes\Phi$ is a finite group of size $|G||\Phi|$
    \item Let $(g,\phi)\in G\ltimes\Phi$. One has 
    \[(g,\phi)^{-1}=(\phi^{-1}(g^{-1}),\phi^{-1})\]
\end{enumerate}
\end{proposition}

\subsection{Proofs of Knowledge and Identification Schemes}\label{sec:id-schemes}

Roughly speaking, the idea of the Fiat-Shamir class of signatures is as follows: we interactively convince an `honest' party that we possess a certain secret. We can then transform this interactive paradigm to a non-interactive digital signature scheme by applying the Fiat-Shamir transform. A primary motivation for this approach is that the resulting signature scheme inherits its security at rather low cost from security properties of the underlying interactive scheme - as such, it is necessary for us now to review some of these security notions.

First, let us define exactly what we mean by these interactive proof of knowledge protocols. The idea of communicating a `secret' is neatly captured by the notion of a binary relation; that is, for two sets $\mathcal{W}$ and $\mathcal{S}$, consider a set $\mathcal{R}\subset\mathcal{W}\times\mathcal{S}$. Given a pair $(w,s)\in\mathcal{R}$, we say $s$ is the \textit{statement} and $w$ is the \textit{witness}. In general, for a given statement a party called the `prover' wishes to demonstrate their knowledge of a valid witness (that is, given $s$ we wish to prove that we possess a $w$ such that $(w,s)\in\mathcal{R}$) to a party called the \textit{verifier}. Of course, one can do this trivially by simply revealing the witness, so we add the crucial requirement that \textit{no information about the witness is revealed}.

We refer more or less to this idea when discussing identification schemes, with the caveat that the prover should be able to compute an arbitrary pair of the binary relation. If the prover cannot generate an an arbitrary pair of the binary relation, and instead is to demonstrate his knowledge of some given element of the binary relation, we have instead a `zero-knowledge proof'. A notable class of zero-knowledge proofs are the so-called `sigma protocols'. One can always turn a zero-knowledge proof into an identification scheme by providing the prover with an algorithm capable of generating an arbitrary pair of the binary relation; our definition of identification schemes in fact refers only to those arrived at by transforming a sigma protocol into an identification scheme.

Notice that the idea of a binary relation serves as a neat generalisation of the usual notion of a public and private key pair. The algorithm used by the identification scheme to generate binary relation instances is therefore denoted by {\tt KeyGen}, and produces a pair $(sk,pk)$. We also require, in some sense to be made precise later, that recovering an appropriate witness from a statement is computationally difficult.

\begin{definition}[Identification Scheme]
    Let $\mathcal{R}\subset\mathcal{S}\times\mathcal{P}$ be a binary relation. An identification scheme is a triple of algorithms (\kg,\peg,\vic), where
    \begin{itemize}
        \item \kg{} takes as input a security parameter $n$ and generates a pair $(sk,pk)\in\mathcal{R}$, publishes $pk$, and passes $sk$ to \peg{}
        \item \peg{} is an interactive algorithm initialised with a pair $(sk,pk)\in\mathcal{R}$
        \item \vic{} is an interactive algorithm initialised with a statement $pk\in\mathcal{P}$. After the interaction, \vic{} outputs a decision `Accept' or `Reject'.
    \end{itemize}

    The interaction of \peg{} and \vic{} runs as follows:
    \begin{enumerate}
        \item \peg{} generates a random `commitment' $I$ from the space of all possible commitments $\mathcal{I}$ and sends it to \vic{}
        \item Upon receipt of $I$, \vic{} chooses a `challenge' $c$ from the space of all possible challenges $\mathcal{C}$ at random and sends it to \peg{}
        \item \peg{} responds with a `response' $p$
        \item \vic{} calculates an `Accept' or `Reject' response as a function of $(I,c,p)$ and the statement $pk$.
    \end{enumerate}
    The interaction of \peg{} and \vic{} is depicted in Figure~\ref{fig:id-scheme}.
\end{definition}

    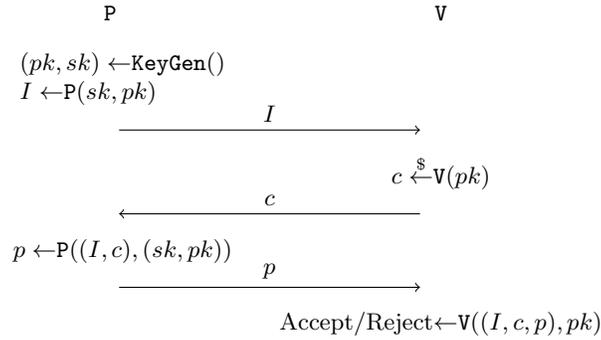
\begin{figure}
    \begin{center}
        \begin{tikzpicture}
            \node (peg) {\peg{}};
            \node (vic) [right=4cm of peg] {\vic{}};
            \node (pick_commit) [below= 0.2cm of peg] {
            \begin{varwidth}{\linewidth}
            \begin{algorithmic}
                \State $(pk,sk)\gets$\kg{}()
                \State $I\gets$\peg{}$(sk,pk)$
            \end{algorithmic}
            \end{varwidth}
            };
            \node (peg_commit) [below= 0.1cm of pick_commit] {};
            \node (vic_commit) [right=4cm of peg_commit] {};
            \draw[->] (peg_commit.east) -- (vic_commit.west) node[midway,above] {$I$};
            \node (vic_select) [below= 0.1cm of vic_commit] {
            \begin{varwidth}{\linewidth}
            \begin{algorithmic}
                \State $c\getsr$\vic{}$(pk)$
            \end{algorithmic}
            \end{varwidth}
            };
            \node (vic_request) [below=0.1cm of vic_select] {};
            \node (peg_request) [left=4cm of vic_request] {};
            \draw[<-] (peg_request.east) -- (vic_request.west) node[midway,above] {$c$};
            \node (peg_calculate) [below = 0.1cm of peg_request] {
            \begin{varwidth}{\linewidth}
            \begin{algorithmic}
                \State $p\gets$\peg$((I,c),(sk,pk))$
            \end{algorithmic}
            \end{varwidth}
            };
            \node (peg_confirm) [below=0.1cm of peg_calculate] {};
            \node (vic_confirm) [right=4cm of peg_confirm] {};
            \draw[->] (peg_confirm.east) -- (vic_confirm.west) node[midway,above] {$p$};
            \node (vic_verify) [below=0.1cm of vic_confirm] {
            \begin{varwidth}{\linewidth}
            \begin{algorithmic}
                \State Accept/Reject$\gets$\vic{}$((I,c,p),pk)$
            \end{algorithmic}
                
            \end{varwidth}
            };
        \end{tikzpicture}
    \end{center}
        
    \caption{An identification scheme.}
    \label{fig:id-scheme}
    \end{figure}

\begin{definition}
    Let (\kg{},\peg{},\vic{}) be an identification scheme. The triple $(I,c,p)$ of exchanged values between \peg{} and \vic{} is called a `transcript'; if a prover (resp. verifier) generates $I,p$ (resp $c$) with the algorithm \peg{} (resp. \vic{}), they are called `honest'. An identification scheme is `complete' if a transcript generated by two honest parties is always accepted by the verifier.
\end{definition}

Turning our attention to the security of identification protocols, let us define the framework we wish to work with. As we will see later, it suffices for signature security to only consider identification schemes for which we have an honest verifier - in other words, it suffices to consider only a cheating prover. Let us do so in the form of the following attack games, which are \cite[Attack~Game~18.1]{boneh2020graduate} and \cite[Attack~Game~18.2]{boneh2020graduate} respectively.

\begin{definition}[Direct Attack Game]
    Let {\tt ID}=(\kg{},\peg{},\vic{}) be an identification scheme and $\mathcal{A}$ be an adversary. Consider the following game:
    \begin{enumerate}
        \item The challenger obtains $(sk,pk)\leftarrow$\kg{} and passes $pk$ to $\mathcal{A}$.
        \item The adversary interacts with the challenger who generates responses with \vic{}. At the end, the challenger outputs `Accept/Reject' as a function of the generated transcript and $pk$; the adversary wins the game if \vic{} outputs `Accept'.
    \end{enumerate}
    The Direct Attack game is depicted in Figure~\ref{fig:dir-att}. We denote the advantage of the adversary in this game with {\tt ID} as the challenger by {\tt dir-adv}($\mathcal{A}$,{\tt ID}).  
\end{definition}

\begin{figure}
    \centering
    \begin{tikzpicture}
        \node[draw,inner sep=1.5cm] (a) {$\mathcal{A}$};
        \node at (-6.5,1.2) (com) {};
        \node at (-6.6,2) (kg) {$(sk,pk)\leftarrow \kg{}()$};

        \node at (0,2) (connecter) {};
        \node at (0,1.5) (pk_r) {};

        \draw[->] (kg.east) .. controls +(right:7mm) and +(up:7mm) .. (pk_r) node[midway,above] {$pk$};
                
        \draw[<-] (com) -- (-1.5,1.2) node[midway,above] {$I^*$};
        \node at (-6.5,0) (ch) {$c\leftarrow${\tt V}$(pk)$};
        \draw[->] (ch) -- (-1.5,0) node[midway,above] {$c$};
        \node at (-6.5,-1.2) (pf) {$d\leftarrow${\tt V}$((I^*,c,p^*),pk)$};
        \draw[<-] (pf) -- (-1.5, -1.2) node[midway,above] {$p^*$};
        \draw (-5, -1.75) rectangle (-8.25, 2.25);

        \node[below =6mm of pf] (dec) {$d$};
        \draw[->] (pf) -- (dec);
        
    \end{tikzpicture}
    \caption{The direct attack game.}
    \label{fig:dir-att}

\end{figure}
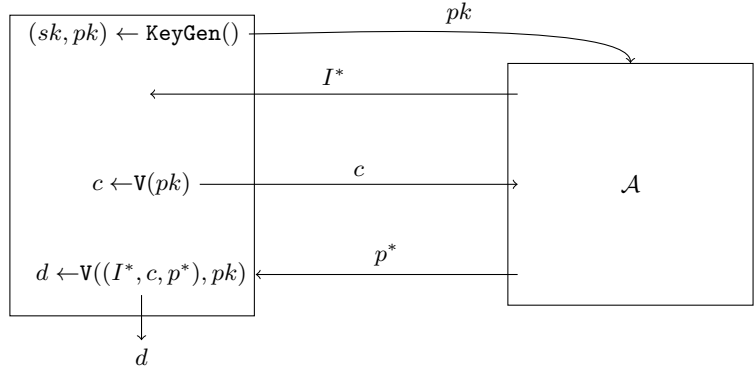

\begin{definition}[Eavesdropping Attack]
    Let {\tt ID}=(\kg{},\peg{},\vic{}) be an identification scheme and $\mathcal{A}$ be an adversary. Consider the following game:
    \begin{enumerate}
        \item The challenger obtains $(sk,pk)\leftarrow$\kg{} and passes $pk$ to $\mathcal{A}$.
        \item The adversary enters into an `eavesdropping' phase, whereby they can request honestly-generated transcripts from a transcript oracle $\mathcal{T}$ possessing the same $(sk,pk)$ pair generated in the previous step.
        \item The adversary interacts with the challenger who generates responses with \vic{}. At the end, the challenger outputs `Accept/Reject' as a function of the generated transcript and $pk$; the adversary wins the game if \vic{} outputs `Accept'.
    \end{enumerate}
    The Eavesdropping Attack game is depicted in Figure~\ref{fig:eav-att}. We denote the advantage of the adversary in this game with {\tt ID} as the challenger by {\tt eav-adv}($\mathcal{A}$,{\tt ID}).
\end{definition}

\begin{figure}
    \centering
    \begin{tikzpicture}
        \node[draw,inner sep=1.5cm] (a) {$\mathcal{A}$};
        \node at (-5.5,1.2) (com) {};
        \node at (-5.6,2) (kg) {$(sk,pk)\leftarrow \kg{}()$};

        \node at (0,2) (connecter) {};
        \node at (0,1.5) (pk_r) {};

        \draw[->] (kg.east) .. controls +(right:7mm) and +(up:7mm) .. (pk_r) node[midway,above] {$pk$};
                
        \draw[<-] (com) -- (-1.5,1.2) node[midway,above] {$I^*$};
        \node at (-5.5,0) (ch) {$c\leftarrow${\tt V}$(pk)$};
        \draw[->] (ch) -- (-1.5,0) node[midway,above] {$c$};
        \node at (-5.5,-1.2) (pf) {$d\leftarrow${\tt V}$((I^*,c,p^*),pk)$};
        \draw[<-] (pf) -- (-1.5, -1.2) node[midway,above] {$p^*$};
        \draw (-4, -1.75) rectangle (-7.25, 2.25);

        \node[below =6mm of pf] (dec) {$d$};
        \draw[->] (pf) -- (dec);

        \node[draw, inner sep=0.75cm] at (4, 0) (or) {$\mathcal{T}$};
        \draw[->] (kg.east) .. controls +(right:20mm) and +(up:16mm) .. (4,0.75) node[midway,above] {$(sk,pk)$};

        \draw[->] (1.5,0.6) -- (3.25,0.6) node[midway,above]{request};
        \draw[<-] (1.5,-0.6) -- (3.25,-0.6) node[midway,above]{$(\hat{I},\hat{c},\hat{p})$};
        
    \end{tikzpicture}
    \caption{The eavesdropping attack game.}
    \label{fig:eav-att}
\end{figure}
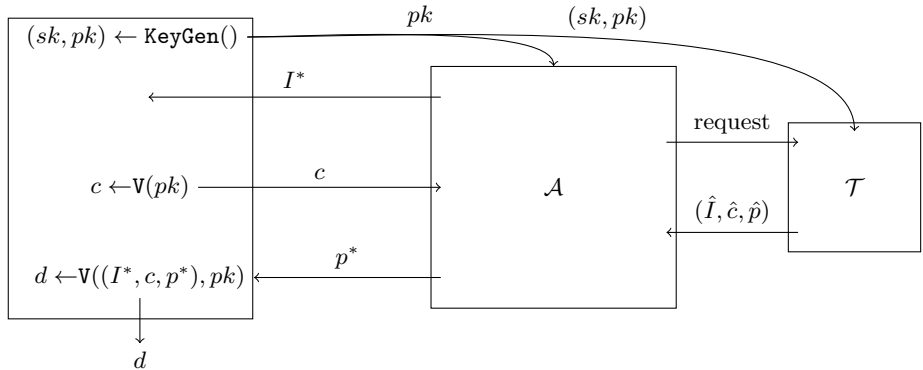

In practice, given a concrete identification scheme it is possible to bound the advantage of an adversary in these games provided one can prove the following two properties hold for the identification scheme:

\begin{definition}\label{def:knowledge}
    Let (\kg{},\peg{},\vic{}) be an identification scheme.
    \begin{itemize}
        \item The scheme has `special soundness' if two transcripts with the same commitment and different challenges allow recovery of the witness $sk$; that is, if $(I,c,p),(I,c^*,p^*)$ are two transcripts generated with $(sk,pk)\leftarrow$\kg, there is an efficient algorithm taking these transcripts as input that returns $sk$.
        \item The scheme has `special honest verifier zero knowledge' if, given a statement $pk$ and a challenge $c$, there is an efficient algorithm to generate a passing transcript $(I^*,c,p^*)$ with the same distribution as a legitimately generated transcript.
    \end{itemize}
\end{definition}

Before moving on there is one final security notion to explore. Notice that if the underlying binary relation of an identification scheme is such that one can easily recover a valid witness from the public statement, an adversary can easily succeed in either of the above games simply by honestly generating the proof $p$ with the appropriate value of $sk$. We have loosely discussed the notion that recovering a witness should therefore be difficult; it is nevertheless so far not clear how precisely this difficulty is accounted for. In fact, there are a number of ways to get round this. For our purposes, and in our application of the Fiat-Shamir transform, we will invoke the system outlined in \cite[Section~19.6]{boneh2020graduate}. The idea is basically thus: provided the properties in Definition~\ref{def:knowledge} hold, it is possible to set up the security proof such that all the difficulty of recovering a witness is `priced in' to the key generation algorithm. Again, we will need a precise definition to make this rigorous later on: the following is \cite[Attack~Game~19.2]{boneh2020graduate}

\begin{definition}[Inversion Attack Game]
    Let \kg{} be a key generation algorithm for a binary relation $\mathcal{R}\subset\mathcal{S}\times\mathcal{P}$ and $\mathcal{A}$ be an adversary. Consider the following game:
    \begin{enumerate}
        \item A pair $(sk,pk)$ is generated by running \kg{}, and the value $pk$ is passed to the adversary $\mathcal{A}$.
        \item $\mathcal{A}$ outputs some $\hat{sk}\in\mathcal{S}$. The adversary wins if $(\hat{sk},pk)\in\mathcal{R}$.
    \end{enumerate}
    We denote the advantage of the adversary in this game with {\tt kg} as the challenger by {\tt inv-adv}($\mathcal{A}$,{\tt kg}).
\end{definition}

\subsection{Signature Schemes}
Recall that a `signature scheme' is a triple of algorithms ({\tt KeyGen, Sg, Vf}), where {\tt KeyGen()} outputs a private-public key pair $(sk,pk)$ upon input of a security parameter. For some space of messages $\mathcal{M}$, {\tt Sg} takes as input $sk$ and some $m\in\mathcal{M}$, producing a `signature' $\sigma$. {\tt Vf} takes as input $pk$ and a pair $(m,\sigma)$, and outputs either `Accept' or `Reject'. We have the obvious correctness requirement that for a key pair $(sk,pk)$ generated by \kg{} we can expect, for any $m\in\mathcal{M}$, that one has  
\[\text{{\tt Vf}}(pk,(m,\text{{\tt Sg}}(sk,m)))=\text{Accept}\]

The security of a signature scheme is defined with respect to the following attack game, which is \cite[Attack Game~13.1]{boneh2020graduate} (but is widely available). 

\begin{definition}[Chosen Message Attack]
    Let {\tt S}=(\kg{},{\tt Sg},{\tt Vf}) be a signature scheme and $\mathcal{A}$ be an adversary. Consider the following game:
    \begin{enumerate}
        \item The challenger obtains $(sk,pk)\leftarrow$\kg{} and passes $pk$ to $\mathcal{A}$.
        \item The adversary enters into an `querying' phase, whereby they can obtain signatures $\sigma_i={\tt Sg}(sk,m_i)$ from the challenger, for the adversary's choice of message $m_i$. The total number of messages queried is denoted $Q$.
        \item The adversary submits their attempted forgery - a message-signature pair $(m^*,\sigma^*)$ - to the challenger. The challenger outputs {\tt Vf}$(pk,(m^*,\sigma^*))$; the adversary wins if this output is `Accept'.
    \end{enumerate}
    The Chosen Message Attack game is depicted in Figure~\ref{fig:cma-att}. We denote the advantage of the adversary in this game with {\tt S} as the challenger by {\tt cma-adv}($\mathcal{A}$,{\tt S}).
\end{definition}

\begin{figure}
    \centering
    \begin{tikzpicture}
        \node[draw,inner sep=1.5cm] (a) {$\mathcal{A}$};
        \node at (-6.5,0.6) (com) {};
        \node at (-6.6,2) (kg) {$(sk,pk)\leftarrow${\tt KeyGen}};

        \node at (-6.5,1) () {for $1\leq i\leq Q$};
        
        \node at (0,2) (connecter) {};
        \node at (0,1.5) (pk_r) {};

        \draw[->] (kg.east) .. controls +(right:7mm) and +(up:7mm) .. (pk_r) node[midway,above] {$pk$};
                
        \draw[<-] (com) -- (-1.5,0.6) node[midway,above] {$m_i$};
        \node at (-6.5,0) (ch) {$\sigma_i\leftarrow${\tt Sg}$(sk,m_i)$};
        \draw[->] (ch) -- (-1.5,0) node[midway,above] {$\sigma_i$};
        \node at (-6.5,-1.2) (pf) {$d\leftarrow${\tt Vf}$(pk,(m^*,\sigma^*))$};
        \draw[<-] (pf) -- (-1.5, -1.2) node[midway,above] {$(m^*,\sigma^*)$};
        \draw (-5, -1.75) rectangle (-8.25, 2.25);

        \node[below =6mm of pf] (dec) {$d$};
        \draw[->] (pf) -- (dec);
        
    \end{tikzpicture}
    \caption{The chosen message attack game.}
    \label{fig:cma-att}
\end{figure}
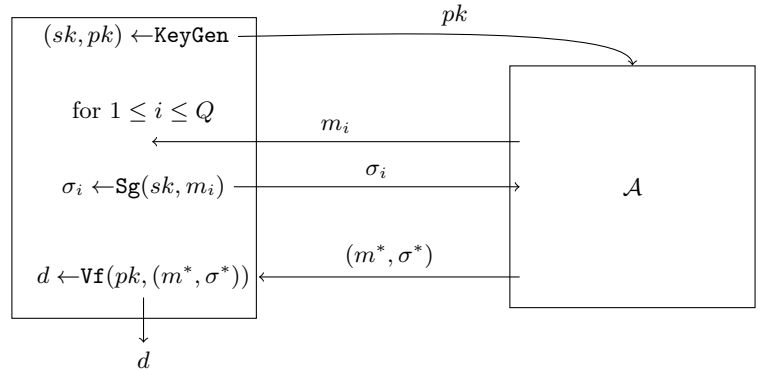

A signature scheme {\tt S} for which {\tt cma-adv}$(\mathcal{A},${\tt S}$)$ is bounded favourably\footnote{`Favourably' here usually means as a negligible function of a security parameter.} from above for any efficient adversary $\mathcal{A}$ is sometimes called {\tt euf-cma} secure, or `existentially unforgeable under chosen message attacks'.

It remains to briefly define the well-known notion of the Fiat-Shamir transform, initially presented in \cite{fiat1987prove}:
\begin{definition}[Fiat-Shamir]
    Let {\tt ID}=(\kg{},\peg{},\vic{}) be an identity scheme with commitment space $\mathcal{I}$ and $\mathcal{C}$. We define a signature scheme FS({\tt ID})=({\tt KeyGen},{\tt Sg},{\tt Vf}) on the message space $\mathcal{M}$ given access to a public function $H:\mathcal{M}\times\mathcal{I}\to\mathcal{C}$:
    \begin{enumerate}
        \item \kg{} is exactly the key generation algorithm of {\tt ID} and outputs a pair $(sk,pk)$, where $pk$ is made public
        \item {\tt Sg} takes as input $m\in\mathcal{M}$ and the key pair $(pk,sk)$ and outputs a signature $(\sigma_1,\sigma_2)$:
        \begin{algorithmic}
            \State $I\gets${\tt P}$((sk,pk))$
            \State $c\gets H(m,I)$
            \State $p\gets${\tt P}$((I,c),(sk,pk))$
            \State $(\sigma_1,\sigma_2)\gets(I,p)$
            \State \Return $(\sigma_1,\sigma_2)$
        \end{algorithmic}
        \item {\tt Vf} takes as input a message-signature pair $(m,(\sigma_1,\sigma_2))$ and outputs a decision $d$, which is `Accept' or `Reject':
        \begin{algorithmic}
            \State $c\gets H(I,\sigma_1)$
            \State $d\gets${\tt V}$((\sigma_1,c,\sigma_2),pk)$
            \State \Return $d$
        \end{algorithmic}
    \end{enumerate}
\end{definition}

Intuitively, we can see that {\tt Sg} is simulating an interactive protocol non-interactively with a call to the function $H$; in order to inherit the security properties of the identification scheme, this function $H$ should have randomly distributed outputs on fresh queries and should be computationally binding - that is, it should be difficult to find a value $I'\neq I$ such that $H(m,I)=H(m,I')$; and given a commitment $c\in\mathcal{C}$ it should be difficult to find a message $m$ and commitment $I\in\mathcal{I}$ such that $H(m,I)=c$. On the other hand, for correctness we need $H$ to be deterministic on previously queried inputs. Such a function is modelled by a hash function under the random oracle model: in this model, it was famously demonstrated in \cite{abdalla2002identification} that a relatively modest security notion for the underlying identification scheme gives strong security proofs for the resulting signature scheme. In our own security proof we use the slightly more textbook exposition presented in \cite{boneh2020graduate}.

\section{A Novel Connection to a Group Action}
Our first task is to demonstrate the existence of the claimed group action, for any finite group. A very similar structure was outlined in \cite{battarbee2022subexponential} - with the important distinction that \textit{semi}groups are insisted upon. Indeed, it turns out that allowing invertibility changes the structure in a way that we shall outline below.

\begin{definition}
        Let $G$ be a finite group, and $\Phi\leq Aut(G)$. Fix some $(g,\phi)\in G\ltimes \Phi$. For any $x\in\ZZ$, the function $s_{g,\phi}:\ZZ\to G$ is defined as the group element such that 
    \[(g,\phi)^x=(s_{g,\phi}(x),\phi^x)\]
\end{definition}

The group action of interest arises from the study of the set $\{s_{g,\phi}(i):i\in\ZZ\}$. Certainly $1\in\{s_{g,\phi}(i):i\in\ZZ\}$, since there is some $n\in\NN$ such that $(s_{g,\phi}(n),\phi^n)=(g,\phi)^n=(1,id)$, but one cannot immediately deduce that this is the smallest integer for which $s_{g,\phi}$ is $1$. Indeed, even if the order $n$ of $(g,\phi)$ is the smallest integer such that $s_{g,\phi}(n)=1$, we are not necessarily guaranteed that every integer up to $n$ is mapped to a distinct elements of $G$ by $s_{g,\phi}$. Before resolving these questions let us introduce some terminology.

\begin{definition}
    Let $G$ be a finite group, and $\Phi\leq Aut(G)$. Fix some $(g,\phi)\in G\ltimes \Phi$. The set 
    \[\mathcal{X}_{g,\phi}:=\{s_{g,\phi}(i):i\in\ZZ\}\]
    is called the \textit{cycle} of $(g,\phi)$, and its size is called the \textit{period} of $(g,\phi)$.
\end{definition}

In the interest of brevity we will also assume henceforth that by $(g,\phi)$ we mean some pair occurring in a semidirect product group as described above. For any such pair $(g,\phi)$, note that $\mathcal{X}_{g,\phi}$ is not necessarily closed under the group operation - we can, nevertheless, implement addition in the argument of $s_{g,\phi}$ as follows:
\begin{theorem}\label{thm:addition-in-exps}
    Let $i,j\in\ZZ$ and suppose $(g,\phi)\in G\ltimes \Phi$ in the usual way. One has that
    \[\phi^j(s_{g,\phi}(i))s_{g,\phi}(j)=s_{g,\phi}(i+j)\]
\end{theorem}
\begin{proof}
    Following the definitions one has 
    \begin{align*}
        (s_{g,\phi}(i+j),\phi^{i+j}) &= (g,\phi)^{i+j} \\
        &= (g,\phi)^i(g,\phi)^j \\
        &= (s_{g,\phi}(i),\phi^i)(s_{g,\phi}(j),\phi^j) \\
        &= (\phi^j(s_{g,\phi}(i))s_{g,\phi}(j),\phi^{i+j})
    \end{align*}    
\qed
\end{proof}

Put another way, we can use integers to map $\mathcal{X}_{g,\phi}$ to itself. This idea is sufficiently important to earn its own notation:
\begin{definition}\label{def:step}
    Let $i\in\ZZ$. The function $\ast:\ZZ\times \mathcal{X}_{g,\phi}\to \mathcal{X}_{g,\phi}$ is given by 
    \[i\ast s_{g,\phi}(j):=\phi^j(s_{g,\phi}(i))s_{g,\phi}(j)\]
\end{definition}

We have seen that $i\ast s_{g,\phi}(j)=s_{g,\phi}(i+j)$; accordingly, we pronounce the $\ast$ symbol as `step'. An immediate consequence is the presence of some degree of `looping' behaviour; that is, supposing $s_{g,\phi}(n)=1$ for some $n\in\ZZ$, one has 
\begin{align*}
    s_{g,\phi}(n+1)=1\ast s_{g,\phi}(n) &= 1\ast 1 \\
    &= \phi(1)s_{g,\phi}(1) \\
    &= s_{g,\phi}(1)
\end{align*}

Generalising this idea we get a more complete picture of the structure of the cycle.

\begin{theorem}\label{cycle-structure}
    Let $G$ be a finite group and $\Phi\leq Aut(G)$ an automorphism subgroup. Fix $(g,\phi)\in G\ltimes Aut(G)$, and let $n$ be the smallest positive integer for which $s_{g,\phi}(n)=1$. One has that $|\mathcal{X}_{g,\phi}|=n$, and
    \[\mathcal{X}_{g,\phi}=\{1,g,...,s_{g,\phi}(n-1)\}\]
\end{theorem}
\begin{proof}
    First, let us demonstrate that the values $1=s_{g,\phi}(0),s_{g,\phi}(1),...,s_{g,\phi}(n-1)$ are all distinct. Suppose to the contrary that there exists $0\leq i<j\leq n-1$ such that $s_{g,\phi}(i)=s_{g,\phi}(j)$; then some positive $k< n$ must be such that $i+k=j$. In other words:
    \begin{align*}
        i\ast s_{g,\phi}(k)=s_{g,\phi}(j) &\Rightarrow \phi^i(s_{g,\phi}(k))s_{g,\phi}(i)=s_{g,\phi}(j) \\
        &\Rightarrow \phi^i(s_{g,\phi}(k)) =1 \\
        &\Rightarrow s_{g,\phi}(k) = 1
    \end{align*}
    which is a contradiction, since $k<n$. It remains to show that every integer is mapped by $s_{g,\phi}$ to one of these $n$ distinct values - but this is trivial, since we can write any integer $i$ as $kn+j$ for some integer $k$ and $0\leq j<n$. It follows that 
    \[s_{g,\phi}(i)=s_{g,\phi}(j)\]
    where $s_{g,\phi}(j)$ is one of the $n$ distinct values.
    \qed
\end{proof}

It follows that we can write $i\ast s_{g,\phi}(j)=s_{g,\phi}(i+j\mod n)$. In fact, the latter part of the above argument demonstrates something slightly stronger: not only is every integer mapped to one of $n$ distinct values by $s_{g,\phi}$, but every member of a distinct residue class modulo $n$ is mapped to the \textit{same} distinct value. It is this basic idea that gives us our group action. 

\begin{theorem}\label{thm:gp-action}
    Let $G$ be a finite group and $\Phi\leq Aut(G)$. Fix a pair $(g,\phi)\in G\ltimes Aut(G)$, and let $n$ be the smallest positive integer such that $s_{g,\phi}(n)=1$. Define the function as
    \begin{align*}
        \circledast:\quad &\ZZ_n\times \mathcal{X}_{g,\phi}\to \mathcal{X}_{g,\phi} \\
        &\modd{i}\circledast s_{g,\phi}(j) = i\ast s_{g,\phi}(j)
    \end{align*}
    The tuple $(\ZZ_n,\mathcal{X}_{g,\phi}, \circledast)$ is a free, transitive group action.
\end{theorem}
\begin{proof}
    First, let us see that $\circledast$ is well-defined. Suppose $i\cong j\mod n$, then $i=j+kn$ for some $k\in\ZZ$. For some arbitrary $\mathcal{X}_{g,\phi}$, say $s_{g,\phi}(l)$ for $0\leq l< n$, one has
    \begin{align*}
        i\ast s_{g,\phi}(l) &= (j+kn)\ast s_{g,\phi}(l) \\
        &= j\ast s_{g,\phi}(l+kn) \\
        &= j\ast s_{g,\phi}(l)
    \end{align*}
    We also need to verify that the claimed tuple is indeed a group action. In order to check that the identity in $\ZZ_n$ fixes each $\mathcal{X}_{g,\phi}$, by the well-definedness just demonstrated, it suffices to check that $0\ast s_{g,\phi}(l)=s_{g,\phi}(l)$ for each $0\leq l< n$ - which indeed is the case. For the compatibility of the action with modular addition, note that for $0\leq i,j,k<n-1$ one has
    \begin{align*}
        \modd{k}\circledast(\modd{j}\circledast s_{g,\phi}(i)) &= \modd{k}\circledast s_{g,\phi}(i+j\mod n)\\
        &= s_{g,\phi}(i+j+k\mod n) \\
        &= \modd{j+k}\circledast s_{g,\phi}(i)
    \end{align*}
    as required. It remains to check that the action is free and transitive. First, suppose $\modd{i}\in\ZZ_n$ fixes each $s_{g,\phi}(j)\in\mathcal{X}_{g,\phi}$. By the above we can assume without loss of generality that $0\leq i<n-1$, and we have $\phi^j(s_{g,\phi}(i))s_{g,\phi}(j)=s_{g,\phi}(j)$. It follows that $s_{g,\phi}(i)=1$, so we must have $i=0$ as required. For transitivity, for any pair $s_{g,\phi}(i),s_{g,\phi}(j)$ we have $\modd{j-i}\circledast s_{g,\phi}(i)=s_{g,\phi}(j)$, and we are done.
\qed
\end{proof}

Recalling that the set $\mathcal{X}_{g,\phi}$ and the period $n$ are a function of the pair $(g,\phi)$, we have actually shown the existence of a large family of group actions. Nevertheless, we have only really shown the existence of the crucial parameter $n$ - it is not necessarily clear how this value should be calculated. With this in mind let us conclude the section with a step in this direction:
\begin{theorem}\label{thm:n-div}
    Fix a pair $(g,\phi)\in G\ltimes Aut(G)$. Let $n$ be the smallest integer such that $s_{g,\phi}(n)=1$, then $n$ divides the order of the pair $(g,\phi)$ as a group element in $G\ltimes Aut(G)$.
\end{theorem}
\begin{proof}
    Suppose $m=ord((g,\phi))$. Certainly $s_{g,\phi}(m)=1$, and by definition one has $m\geq n$. We can therefore write $m=kn+l$, for $k\in\NN$ and $0\leq l< n$.
    It is not too difficult to verify that $s_{g,\phi}(x)=\phi^{x-1}(g)...\phi(g)g$ for any $x\in\NN$. It follows that
    \[s_{g,\phi}(m)=\phi^{kn}(s_{g,\phi}(l))\phi^{(k-1)n}(s_{g,\phi}(n))...\phi^n(s_{g,\phi}(n))s_{g,\phi}(n)\]
    Since $s_{g,\phi}(m)=s_{g,\phi}(n)=1$, we must have $s_{g,\phi}(l)=1$. But $l<n$ and so $l=0$ by the minimality of $n$, which in turn implies that $n|m$ as required.
\qed
\end{proof}

\subsection{Semidirect Discrete Logarithm Problem}\label{sec:sdlp}
Given a group $G$ and a pair $(g,\phi)\in G\ltimes Aut(G)$, observe that as a consequence of Theorem~\ref{thm:addition-in-exps} and Definition~\ref{def:step}, for any two integers $i,j\in\NN$ we have that $s_{g,\phi}(i+j)=j\ast s_{g,\phi}(i)=i\ast s_{g,\phi}(j)$. A Diffie-Hellman style key exchange immediately follows\footnote{Historically speaking, the key exchange predates the more abstract treatment in this paper.}; indeed, a key exchange based on this idea first appears in \cite{habeeb2013public}, and is known as Semidirect Product Key Exchange. In the same way that the security of Diffie-Hellman key exchange is related to the security of the Discrete Logarithm Problem, to understand the security of Semidirect Product Key Exchange we should like to study the difficulty of the following task:

\begin{definition}[Semidirect Discrete Logarithm Problem]
    Let $G$ be a finite group, and let $(g,\phi)\in G\times Aut(G)$. Suppose, for some $x\in\NN$, that one is given $(g,\phi),s_{g,\phi}(x)$; the Semidirect Discrete Logarithm Problem (SDLP) with respect to $(g,\phi)$ is to recover the integer $x$.
\end{definition}

The complexity of SDLP is relatively well understood, in large part due to the connection with group actions highlighted above. We will see later on that the security game advantages for our identification and signature schemes can be bounded in terms of the advantage of an adversary in solving SDLP; indeed, for the SDLP attack game defined in the obvious way, we write the advantage of an adversary {\tt sdlp-adv}($\mathcal{A}$,$(g,\phi)$).

Before we move on to study the signature schemes resulting from each group action we note that the convention in the area is to restrict a finite group $G$ to be a finite, \textit{non-abelian} group $G$. This was in part to preclude a trivial attack on the related key exchange for a specific choice of $\phi$ - nevertheless, throughout the rest of the paper we adopt this convention.

\section{{\tt SPDH-Sign}}

\subsection{An Identification Scheme}
Recall that our strategy is to set up an honest-verifier identification scheme, to which we can apply the well-known Fiat-Shamir heuristic and obtain strong security guarantees in the ROM. The central idea of this identification scheme is as follows: suppose we wish to prove knowledge of some secret $\ZZ_n$ element, say $\modd{s}$. We can select an arbitrary element of $\mathcal{X}_{g,\phi}$, say $X_0$, and publish the pair $X_0,X_1:=\modd{s}\circledast X_0$. An honest party wishing to verify our knowledge of the secret $\modd{s}$ might invite us to commit to some group element $\modd{t}$, for $\modd{t}$ sampled uniformly at random from $\ZZ_n$. We can do this by sending the element $I=\modd{t}\circledast X_0$ - note that as a consequence of the free and transitive properties, $\modd{t}$ is the unique group element such that $I=\modd{t}\circledast X_0$. However, with our knowledge of the secret $\modd{s}$ and the commitment $\modd{t}$, we can calculate the element $\modd{p}=\modd{t-s}$ such that $\modd{p}\circledast X_1=I$, where this equation holds by the group action axioms: one has $\modd{t-s}\circledast (\modd{s}\circledast X_0)=\modd{t}\circledast X_0=I$.

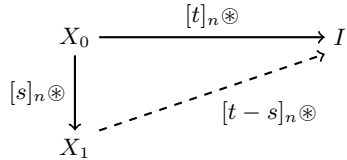
\begin{figure}

\begin{center}
    \begin{tikzpicture}

\centering
\node (s0) {$X_0$};
\node (s1) [below=of s0] {$X_1$};
\node (comm) [right=3cm of s0] {$I$};

\draw[thick, ->] (s0.south) -- (s1.north) node[midway,left]  {$\modd{s}\circledast $};
\draw[thick, ->] (s0.east) -- (comm.west) node[midway,above] {$\modd{t}\circledast$};
\draw[dashed, thick, ->] (s1.north east) -- (comm.south west) node[midway,below right] {$\modd{t-s}\circledast$};

\end{tikzpicture}
\end{center}

\caption{Paths to the commitment.}
\label{fig:comm-paths}
\label{graph-int}
\end{figure}

Interpreted graph-theoretically (as depicted in Figure~\ref{graph-int}), an honest verifier can ask to see one of two paths to the commitment value. Consider a dishonest party attempting to convince the verifier that they possess the secret $\modd{s}$. In attempting to impersonate the honest prover, our dishonest party can generate their own value of $\modd{t}$, and so can certainly provide the correct path in one of the two scenarios. Assuming, however, that recovering the appropriate group element is difficult, without knowledge of the secret $\modd{s}$ this party succeeds in their deception with low probability.

This intuition gives us the following non-rigorous argument of security in the framework described in Section~\ref{sec:id-schemes}. First, recall that we are in the honest verifier scenario, and so a challenge bit $c$ will be $0$ with probability $1/2$, in which case a cheating prover succeeds with probability 1. Supposing that $\varepsilon$ is the probability of successfully recovering the value $\modd{t-s}$, it follows that a cheating prover succeeds with probability $(1+\varepsilon)/2$ - that is, with probability larger than $1/2$. We can quite easily counter this by requiring that $N$ instances are run at the same time. In this case, if $N$ zeroes are selected the prover wins with probability $1$ by revealing their dishonestly generated values of $\modd{t}$ - otherwise, they are required to recover at least 1 value of $\modd{t-s}$. Assuming for simplicity that the probability of doing so remains consistent regardless of the number of times such a value is to be recovered, since the honest verifier selects their challenges uniformly at random the cheating prover succeeds with probability
\[\frac{1}{2^N}+\sum_{i=1}^{2^N-1}\frac{\varepsilon}{2^N}=\frac{1}{2^N}+\varepsilon\frac{2^{N}-1}{2^N}\]
which tends to $\varepsilon$ as $N\to\infty$.

The actual proof of security operates within the security games defined in the preliminaries. As a step towards this formalisation, we need to specify the binary relation our identification scheme is based on. Choose some finite non-abelian group $G$: given a fixed pair $(g,\phi)\in G\times Aut(G)$ we are interested, by Theorem~\ref{thm:gp-action}, in a subset $\mathcal{R}$ of $\ZZ_n,\mathcal{X}_{g,\phi}$, where $n$ is the smallest integer such that $s_{g,\phi}(n)=1$. In fact, legislating for $N$ parallel executions of the proof of knowledge, to each tuple $(X_1,...,X_N)$ is associated a binary relation 
\[\mathcal{R}\subset\ZZ_n^N\times\mathcal{X}_{g,\phi}^N\]
where $((\modd{s_1},...,\modd{s_N}),(Y_1,...,Y_N))\in\mathcal{R}$ exactly when $(Y_1,...,Y_N)=(\modd{s_1}\ast X_1,...,\modd{s_N}\ast X_N)$.

With all this in mind let us define our identification scheme. The more rigorous presentation should not distract from the intuition that we describe $N$ parallel executions of the game in Figure~\ref{fig:comm-paths}.

\begin{pro}
Let $G$ be a finite non-abelian group and $(g,\phi)\in G\ltimes Aut(G)$. Suppose also that $n\in\mathbb{N}$ is the smallest integer such that $s_{g,\phi}(n)=1$. The identification scheme {\tt SPDH-ID}$_{g,\phi}(N)$ is a triple of algorithms

\[\text{({\tt KeyGen}$_{g,\phi}$,{\tt P}$_{g,\phi}$,{\tt V}$_{g,\phi}$)}\]

such that
\begin{enumerate}
    \item {\tt KeyGen}$_{g,\phi}$ takes as input some $N\in\mathbb{N}$.
    \begin{algorithmic}
        \State $(X_1,...,X_N)\gets \mathcal{X}_{g,\phi}^N$
        \State $([s_1]_n,...,[s_N]_n)\gets\ZZ_n^N$
        \State $(Y_1,...,Y_N)\gets ([s_1]_n\circledast X_1,...,[s_N]_n\circledast X_N)$   
    \end{algorithmic}
     {\tt KeyGen}$_{g,\phi}$ outputs the public key $((X_1,...,X_N),(Y_1,...,Y_N))$ and passes the secret key $([s_1]_n,...,[s_N]_n)$ to the prover {\tt P}$_{g,\phi}$. The public key and the value of $N$ used is published.
     \item {\tt P}$_{g,\phi}$ and {\tt V}$_{g,\phi}$ are interactive algorithms that work as depicted in Figure~\ref{fig:spdh-id}:
     \begin{figure}
     \begin{center}
        \begin{tikzpicture}
            \node (peg) {\large{\tt P}$_{g,\phi}$};
            \node (vic) [right=4cm of peg] {\large{\tt V}$_{g,\phi}$};
            \node (pick_commit) [below= 0.2cm of peg] {
            \begin{varwidth}{\linewidth}
            \begin{algorithmic}
                \For{$i\gets1,N$}
                    \State$\modd{t_i}\getsr\ZZ_n$
                    \State$I_i\gets\modd{t_i}\circledast X_i$
                \EndFor
                \State $I\gets(I_1,...,I_N)$
            \end{algorithmic}
            \end{varwidth}
            };
            \node (peg_commit) [below= 0.1cm of pick_commit] {};
            \node (vic_commit) [right=4cm of peg_commit] {};
            \draw[->] (peg_commit.east) -- (vic_commit.west) node[midway,above] {$I$};
            \node (vic_select) [below= 0.1cm of vic_commit] {
            \begin{varwidth}{\linewidth}
            \begin{algorithmic}
                \For{$i\gets1,N$}
                    \State$c_i\getsr\{0,1\}$
                \EndFor
                \State $c\gets(c_1,...,c_N)$
            \end{algorithmic}
            \end{varwidth}
            };
            \node (vic_request) [below=0.1cm of vic_select] {};
            \node (peg_request) [left=4cm of vic_request] {};
            \draw[<-] (peg_request.east) -- (vic_request.west) node[midway,above] {$c$};
            \node (peg_calculate) [below = 0.3cm of peg_request] {
            \begin{varwidth}{\linewidth}
            \begin{algorithmic}
                \For{$i\gets1,N$}
                    \If{$c_i=0$}
                    \State$\modd{p_i}\gets\modd{t_i}$
                    \Else
                    \State$\modd{p_i}\gets\modd{t_i-s_i}$
                    \EndIf
                \EndFor
                \State $p\gets(\modd{p_1},...,\modd{p_N})$
            \end{algorithmic}
            \end{varwidth}
            };
            \node (peg_confirm) [below=0.1cm of peg_calculate] {};
            \node (vic_confirm) [right=4cm of peg_confirm] {};
            \draw[->] (peg_confirm.east) -- (vic_confirm.west) node[midway,above] {$p$};
            \node (vic_verify) [below=0.1cm of vic_confirm] {
            \begin{varwidth}{\linewidth}
            \begin{algorithmic}
                \For{$i\gets1,N$}
                    \If{$c_i=0$}
                    \State $V_i\gets\modd{p_i}\circledast X_i$
                    \Else
                    \State $V_i\gets\modd{p_i}\circledast Y_i$
                    \EndIf
                \EndFor
                \State$V\gets(V_1,...,V_N)$
                \State $d\gets I\stackrel{?}{=}V$
                \State \Return $d$
            \end{algorithmic}
                
            \end{varwidth}
            };
        \end{tikzpicture}
    \end{center}
    \caption{{\tt SPDH-ID}}
    \label{fig:spdh-id}
    \end{figure}
\end{enumerate}
\end{pro}

\subsubsection{Security}

In this section we demonstrate that {\tt SPDH-ID} is secure against eavesdropping attacks in the following sense: the advantage of an adversary in the eavesdropping attack game can be bounded by that of the adversary in the SDLP game. First, let us check that the desirable properties of an identification scheme hold:

\begin{theorem}\label{zkp-properties}
    {\tt SPDH-ID} has the following properties:
    \begin{enumerate}
        \item Completeness
        \item Special soundness
        \item Special honest-verifier zero knowledge.
    \end{enumerate}
\end{theorem}
\begin{proof}
Note that in order to prove each of these properties on the $N$-tuples comprising the transcripts generated by {\tt SPDH-ID}, we need to prove that the properties hold for each component of the tuple; but since each component is independent of all the others, it suffices to demonstrate the stated properties for a single arbitrary component. In other words, we show that the stated properties hold when $N=1$, and the general case immediately follows.
\begin{enumerate}
    \item If $b=0$ then $\modd{p}=\modd{t}$, and trivially we are done. If $b=1$ then $\modd{p}=\modd{t-s}$; doing the bookkeeping we get that
    \begin{align*}
        \modd{p}\circledast S_1 &= \modd{p} \circledast (\modd{s}\circledast S_0) \\
        &= (\modd{t-s}\modd{s})\circledast S_0 \\
        &= (\modd{s}\circledast S_0) = I
    \end{align*}
    \item Two passing transcripts with the same commitment are $(I,0,\modd{t})$ and $(I,1,\modd{t-s})$. Labelling the two responses $x^{p_1},x^{p_2}$, we recover the secret as $(x^{p_2})^{-1}(x^{p_1})$.
    \item It suffices to show that one can produce passing transcripts with the same distribution as legitimate transcripts, but without knowledge of $\modd{s}$. We have already discussed how to produce these transcripts; if a simulator samples $\modd{t}$ uniformly at random, then the transcript $(\modd{t}\circledast S_b, b, \modd{t})$ is valid regardless of the value of $b$. Moreover, if $b=0$, trivially the transcripts have the same distribution; if $b=1$, since $\modd{s}$ is fixed and $\modd{t}$ is sampled uniformly at random, the distribution of a legitimate passing transcript is also uniformly random.
\end{enumerate}
\qed
\end{proof}

We are now ready to bound on the security of our identification scheme.

\begin{theorem}\label{thm:id-security}
    Let $G$ be a finite abelian group and let $(g,\phi)\in G\ltimes Aut(G)$. For some $N\in\NN$, consider the identification scheme {\tt SPDH-ID}$_{g,\phi}(N)$ and an efficient adversary $\mathcal{A}$. There exists an efficient adversary $\mathcal{B}$ with $\mathcal{A}$ as a subroutine, such that with $\varepsilon=$ {\tt sdlp-adv}$(\mathcal{B},(g,\phi))$, we have
    \[\text{{\tt eav-adv}}(\mathcal{A},\text{{\tt SPDH-ID}}_{g,\phi}(N))\leq \sqrt{\varepsilon}+\frac{1}{2^N}\]
\end{theorem}
\begin{proof}
    This is just a straightforward application of two results in \cite{boneh2020graduate}. By \cite[Theorem~19.14]{boneh2020graduate}, since {\tt SPDH-ID}$_{g,\phi}(N)$ has honest verifier zero knowledge, there exists an efficient adversary $\mathcal{B'}$ with $\mathcal{A}$ as a subroutine such that 
    \[\text{{\tt eav-adv}}(\mathcal{A},\text{{\tt SPDH-ID}}_{g,\phi}(N))=\text{{\tt dir-adv}}(\mathcal{B'},\text{{\tt SPDH-ID}}_{g,\phi}(N))\]
    Moreover, let 
    \[\delta=\text{{\tt inv-adv}}(\mathcal{B'},\text{{\tt KeyGen}}_{g,\phi})\]
    Since {\tt SPDH-ID}$_{g,\phi}(N)$ has special soundness, \cite[Theorem~19.13]{boneh2020graduate} gives
    \[\text{{\tt dir-adv}}(\mathcal{B},\text{{\tt SPDH-ID}}_{g,\phi}(N))\leq \sqrt{\delta}+\frac{1}{M}\]
    where $M$ is the size of the challenge space. It is easy to see that $M=2^N$; it remains to relate the quantities 
    $\varepsilon$ and $\delta$. We do so eschewing some of the detail since the argument is straightforward; note that by definition of the binary relation underpinning {\tt KeyGen}$_{g,\phi}$, we can think of the inversion attack game as a security game in which one solves $N$ independent SDLP instances in parallel. Call the advantage in this game {\tt N-sdlp-adv}$(\mathcal{B'},(g,\phi))$, and suppose an adversary $\mathcal{B}$ in the standard SDLP attack game runs $\mathcal{B'}$ as an adversary. $\mathcal{B}$ can simply provide $\mathcal{B'}$ with $N$ copies of its challenge SDLP instance, and succeeds whenever $\mathcal{B'}$ does. It follows that $\delta\leq\varepsilon$, and we are done.  
\qed
\end{proof}

\subsection{A Digital Signature Scheme}
It remains now to apply the Fiat-Shamir transform to our identification scheme. Doing so yields the signature scheme claimed in the title of this paper.

\begin{pro}[{\tt SPDH-Sign}]
    Let $G$ be a finite non-abelian group and let $(g,\phi)\in G\times Aut(G)$ be such that $n$ is the smallest integer for which $s_{g,\phi}(n)$=1. For any $N\in\NN$ and message space $\mathcal{M}$, suppose we are provided a hash function $H:\mathcal{X}_{g,\phi}^N\times\mathcal{M}\to\{0,1\}^N$. We define the signature scheme 
    \[\text{{\tt SPDH-Sign}}_{g,\phi}(N)=(\text{{\tt KeyGen}, {\tt Sg}, {\tt Vf}})\] 
    as in Figure~\ref{fig:spdh-sign}.
\end{pro}

\begin{figure}
    \centering
    
    \begin{tabular}{lll}
    \begin{minipage}[t]{0.44\textwidth}
    {\tt KeyGen}$(N)$:
        \begin{algorithmic}
        \For{$i\gets 1,N$}
            \State $X_i\getsr\mathcal{X}_{g,\phi}$
            \State $\modd{s_i}\getsr\ZZ_n$
            \State $Y_i\gets\modd{s_i}\circledast X_i$
        \EndFor
        \State $sk\gets(\modd{s_1},...,\modd{s_N})$
        \State $pk\gets((X_1,...,X_N),(Y_1,...,Y_N))$
        \State \Return $(sk,pk)$
        \end{algorithmic}
    \end{minipage}
    & 
    \begin{minipage}[t]{0.28\textwidth}
    {\tt Sg}$(m,(sk,pk))$:
        \begin{algorithmic}
        \For{$i\gets1,N$}
            \State $\modd{t_i}\getsr\ZZ_n$
            \State $I_i\gets\modd{t_i}\circledast X_i$
        \EndFor
        \State $I\gets(I_1,...,I_N)$
        \State $c\gets H(I,m)$
        \For{$i\gets1,N$}
            \If{$c_i=0$}
            \State $p_i\gets \modd{t_i}$
            \Else
            \State $p_i\gets\modd{t_i-s_i}$
            \EndIf
        \EndFor
        \State $p\gets (p_1,...,p_N)$
        \State ($\sigma_1,\sigma_2) \gets (I,p)$
        \State \Return $(\sigma_1,\sigma_2)$
        \end{algorithmic}
    \end{minipage}
    &
    \begin{minipage}[t]{0.26\textwidth}
    {\tt Vf}$(m,(\sigma_1,\sigma_2),pk)$:
        \begin{algorithmic}
        \State $c\gets H(\sigma_1,m)$
        \For{$i\gets1,N$}
            \If {$c_i=0$}
            \State $V_i\gets p_i\circledast X_i$
            \Else
            \State $V_i\gets p_i\circledast Y_i$
            \EndIf
        \EndFor
        \State $V\gets(V_1,...,V_N)$
        \State $d\gets V\stackrel{?}{=}I $
        \State \Return $d$
        \end{algorithmic}
    \end{minipage}
\end{tabular}

    \caption{{\tt SPDH-Sign}}
    \label{fig:spdh-sign}
\end{figure}

It is easy to see that given the identification scheme {\tt SPDH-ID}$_{g,\phi}(N)$, the signature scheme {\tt SPDH-Sign}$_{g,\phi}(N)$ is exactly FS({\tt SPDH-ID}$_{g,\phi}(N)$). Before we can use this fact to prove the security of the signature, we require that the hash function gives outputs distributed at `random', in some sense. This is accounted for by the `Random Oracle Model': every time we wish to compute the hash function $H$, we suppose that an oracle function of the appropriate dimension selected at random is queried. Any party can query the random oracle at any time, and the number of these queries is kept track of. We also note that we do not in this paper account for the quantum-accessible random oracle model required for post-quantum security - equivalent security proofs in the quantum-accessible random oracle model are provided, for example, in \cite{beullens2019csi}.

With this heuristic in place we can prove the security of our signature scheme relative to SDLP with a simple application of \cite[Theorem~19.15]{boneh2020graduate} and its corollaries:

\begin{theorem}\label{thm:sig-sec}
    Let $G$ be a finite non-abelian group; $(g,\phi)\in G\ltimes Aut(G)$; and $n\in\NN$ be the smallest integer such that $s_{g,\phi}(n)=1$. Consider the chosen message attack game in the random oracle model, where $Q_s$ is the number of signing queries made and $Q_{ro}$ is the number of random oracle queries. For any efficient adversary $\mathcal{A}$ and $N\in\NN$, there exists an efficient adversary $\mathcal{B}$ running $\mathcal{A}$ as a subroutine such that the signature scheme {\tt SPDH-Sign}$_{g,\phi}(N)$ has
    \[\delta\leq \frac{Q_s}{n}(Q_s+Q_{ro}+1)+\frac{Q_{ro}}{2^N}+\sqrt{(Q_{ro}+1)\text{{\tt sdlp-adv}}(\mathcal{B},(g,\phi))}\]
    where $\delta=\text{{\tt cma-adv}}^{\text{{\tt ro}}}(\text{{\tt SPDH-Sign}}_{g,\phi}(N),\mathcal{A})$ is the advantage of the signature scheme in the random oracle model version of the chosen message attack game.
\end{theorem}
\begin{proof}
    Applying \cite[Theorem~19.15]{boneh2020graduate} and \cite[Equation~19.21]{boneh2020graduate}, since the underlying identification scheme has honest verifier zero knowledge there is an efficient adversary $\mathcal{B'}$ running $\mathcal{A}$ as a subroutine such that 
    \[\delta\leq \gamma Q_s(Q_s+Q_{ro}+1)+\frac{Q_{ro}}{|\mathcal{C}|}+\sqrt{(Q_{ro}+1)\text{{\tt inv-adv}}(\mathcal{B},\text{{\tt KeyGen}}_{g,\phi})}\]
    where $\gamma$ is the probability that a given commitment value appears in a transcript, and {\tt KeyGen}$_{g,\phi}$ is the key generation algorithm of the underlying identification scheme. Since choosing a random group element corresponds to choosing a random element of $\mathcal{X}_{g,\phi}$, each commitment value in $\mathcal{X}_{g,\phi}$ has probability $1/|\mathcal{X}_{g,\phi}|=1/n$ of being selected. We have already seen in the proof of Theorem~\ref{thm:id-security} that the advantage of an adversary in the inversion attack game against this key generation algorithm is bounded by the advantage in an SDLP attack game, and the result follows.
\qed
\end{proof}

The above theorem provides a concrete estimate on the advantage of an adversary in the chosen message attack game; nevertheless, a plain English rephrasing is a useful reflection on these results. Essentially, we now know that the {\tt euf-cma} security of our signature scheme is reliant on the integer $n$ corresponding to the pair $(g,\phi)$, the size of $N$, and the difficulty of SDLP relative to the pair $(g,\phi)$. We can discount the reliance on $N$, which can be `artificially' inflated as we please; note also that we can intuitively expect the size of $n$ and the difficulty of SDLP for $(g,\phi)$ to be at least somewhat correlated, since a small value of $n$ trivially renders the associated SDLP instance easy by brute force. In essence, then, we have shown that we can expect the signature scheme corresponding to $(g,\phi)$ to be secure provided the associated SDLP instance is difficult.

\section{On the Difficulty of SDLP}

For any finite non-abelian group $G$, we have shown the existence of signature scheme for any pair $(g,\phi)\in G\times Aut(G)$. It is now clear from Theorem~\ref{thm:sig-sec} that if the signature is defined with respect to a pair $(g,\phi)$, SDLP with respect to $(g,\phi)$ should be difficult. In this section we discuss sensible choices of $G$ with respect to this criterion.

As alluded to in the title of this paper we are interested in post-quantum hard instances of SDLP; that is, if an instance of SDLP has a known reduction to a quantum-vulnerable problem we should consider it to be easy. 

There are three key strategies in the literature for addressing SDLP. Two of them, at face value, appear to solve a problem instead related to SDLP: let us explore the gap between the problems below.

\subsection{Dihedral Hidden Subgroup Problem}
It should first be noted that, as with all group action-based cryptography, the Dihredral Hidden Subgroup Problem will be highly relevant. Indeed, we can bound the complexity of SDLP above by appealing to Kuperberg's celebrated quantum algorithm for the Abelian Hidden Shift Problem \cite{kuperberg2005subexponential}, defined as follows:

\begin{definition}
    Let $A$ be an abelian group and $S$ be a set. Consider two injective functions $f,g:A\to S$ such that for some $h\in A$, we have $f(a)=g(a+h)$ for all $a\in A$. We say that the functions $f,g$ `hide' $h$, and the Abelian Hidden Shift Problem is to recover $h$ via queries to $f,g$.
\end{definition}

Adapting an argument seen throughout the literature, but first codified in its modern sense in \cite{childs2014constructing}, gives us the following result.
\begin{theorem}\label{thm:sdlp-comp}
    Let $G$ be a finite non-abelian group and let $(g,\phi)\in G\ltimes Aut(G)$; and $n\in\NN$ be the smallest integer such that $s_{g,\phi}(n)=1$. Given $(g,\phi)$ and a group element $s_{g,\phi}(x)$, there is a quantum algorithm that recovers $x$ in time $2^{\mathcal{O}(\sqrt{\log n})}$.
\end{theorem}
\begin{proof}
    If the relevant abelian group has size $n$ we have the claimed complexity for an abelian hidden shift problem by \cite[Proposition~6.1]{kuperberg2005subexponential}. It suffices to show that one can solve SDLP provided one can solve the abelian hidden shift problem - the argument goes as follows. Define $f,g:\ZZ_n\to\mathcal{X}_{g,\phi}$ by 
    \[f(\modd{z})=\modd{z}\circledast s_{g,\phi}(x)\quad g(\modd{z})=\modd{z}\circledast s_{g,\phi}(1)\]    
    We have for all $\modd{z}\in \ZZ_n$ that
    \begin{align*}
        f(\modd{z}) &= \modd{z}\circledast s_{g,\phi}(x) \\
                    &= \modd{z}\circledast(\modd{x-1}\circledast s_{g,\phi}(1)) \\
                    &= (\modd{z}+\modd{x-1})\circledast s_{g,\phi}(1) \\
                    &= g(\modd{z}+\modd{x-1})\circledast s_{g,\phi}(1)
    \end{align*}
    so $f$ and $g$ hide $\modd{x-1}$, from which $x\in\NN$ can be recovered trivially. \qed
\end{proof}

A small amount of detail is suppressed in the above proof: namely, that we have tacitly assumed knowledge of the quantity $n$. Since the best algorithm for the abelian hidden shift problem is quantum anyway, we need not be reticent to compute $n$ with a quantum algorithm - and since the function $s_{g,\phi}$ is periodic in $n$, certainly such Shor-like techniques are available, such as \cite[Algorithm~5]{childs2010quantum}. On the other hand, the ability to compute $n$ efficiently and classically is both desirable and addressed later in this paper. 

\subsection{Semidirect Computational Diffie-Hellman}
The other major body of work related to the analysis of SDLP addresses the following related problem:

\begin{definition}[Semidirect Computational Diffie-Hellman]
    Let $G$ be a finite abelian group, and let $(g,\phi)\in G\ltimes Aut(G)$. Let $x,y\in\NN$ and suppose we are given the data $(g,\phi),s_{g,\phi}(x)$ and $s_{g,\phi}(y)$. The Semidirect Computational Diffie-Hellman problem (SCDH) is to compute the value $s_{g,\phi}(x+y)$.
\end{definition}

Recall our discussion of Semidirect Product Key Exchange in Section~\ref{sec:sdlp}. Notice that SCDH is, similarly to the role of the classic CDH, precisely the problem of key recovery in Semidirect Product Key Exchange, and moreover that the relationship between SCDH and SDLP is not immediately obvious. Of course, one can solve SCDH if one can solve SDLP, but the converse does not follow \textit{a priori}.

There are two general approaches for solving SCDH:

\paragraph{The Dimension Attack.} The general form of this argument appears in \cite{myasnikov2015linear}; we prefer the slightly more purpose-built exposition of \cite{roman2015linear}. The idea is basically that if our group $G$ can be embedded as a multiplicative subgroup of a finite-dimensional algebra over a field, and if the automorphism $\phi$ can be extended to preserve addition on this algebra, we can solve SCDH for some pair $(g,\phi)$ using Gaussian elimination. 
\paragraph{The Telescoping Attack.} In \cite{brown2021cryptanalysis}, it is noticed that $1\ast s_{g,\phi}(x)=\phi^x(g)s_{g,\phi}(x)$. Since we know $s_{g,\phi}(x)$ we can calculate $1\ast s_{g,\phi}(x)$ and solve for $\phi^x(g)$. In some cases - notably, in the additive structure given in \cite{rahman2022make} - this suffices for recovery of $s_{g,\phi}(x+y)$.

We comment that a method of efficiently converting an SCDH solver to an SDLP solver is not currently known. On the other hand, a recent result of Montgomery and Zhandry \cite{montgomery2022full} shows that a computational problem underpinning SDLP and a computational problem underpinning SCDH\footnote{More precisely, the Vectorisation and Parallelisation problems of Couveignes \cite{couveignes2006hard}, respectively.} are (surprisingly) quantum equivalent. We therefore cautiously conjecture that there exists some efficient quantum method of converting an SCDH solver to an SDLP solver.

\section{A Candidate Group}
We propose the following group of order $p^3$, where $p$ is an odd prime, for use with {\tt SPDH-Sign}.

\begin{definition}
    Let $p$ be an odd prime. The group $G_p$ is defined by 
    \[G_p=\left\{ 
    \begin{pmatrix}
    a & b \\
    0 & 1
    \end{pmatrix}
    : a,b\in\ZZ_{p^2}, a\equiv1\mod p
    \right\}\]
\end{definition}

As discussed in \cite{conrad}, this group is one of two non-abelian groups of order $p^3$ for an odd prime up to isomorphism. It has presentation
\[G_p=\langle x,y:y^p=1,[x,y]=x^p=:z\in Z(G_p),z^p=1\rangle\]
as described in \cite{mahalanobis2011mor}; moreover, its automorphism group is known and has size $(p-1)p^3$ by \cite[Theorem~3.1]{curran2008automorphism}.

With respect to the various matters discussed in this paper, we briefly present the advantages of employing such a group.

\paragraph{Sampling.} Recall that our security proof for {\tt SPDH-Sign} relied heavily on the underlying identification scheme being honest-verifier zero knowledge, which in turn relied on the `fake' transcripts to have the same distribution as honestly generated transcripts. For a pair $(g,\phi)$, it is therefore important to be able to sample uniformly at random from the group $\ZZ_n$, where $n$ is the smallest integer for which $s_{g,\phi}(n)=1$ - in our case, to do so it clearly suffices to compute $n$.

Here we recall Theorem~\ref{thm:n-div}, which tells us basically that, thinking of $(g,\phi)$ as a member of the semidirect product group $G\ltimes Aut(G)$, $n$ must divide the order of $(g,\phi)$. We therefore have the following
\begin{theorem}\label{thm:compute-n}
    Let $(g,\phi)\in G_p\times Aut(G_p)$, where $p$ is an odd prime. Suppose $n$ is the smallest integer for which $s_{g,\phi}(n)=1$. Then 
    \[n\in\{p,p^2,p^3,p^4,p^5,p^6,(p-1),p(p-1),p^2(p-1),p^3(p-1),p^4(p-1),p^5(p-1)\}\]
\end{theorem}
\begin{proof}
    By Theorem~\ref{thm:n-div} we know that $n\vert ord((g,\phi))$, and it is standard that $$ord((g,\phi))\quad\vert\quad|G_p\ltimes Aut(G)|$$. We know from the discussion at the outset of this section that $|G_p|=p^3$ and $|Aut(G_p)|=p^3(p-1)$. It follows that $n|p^3p^3(p-1)$. Since $p$ is prime, and assuming that $(g,\phi)$ is not the identity, the claimed set is a complete list of divisors of $p^6(p-1)$ - excluding $p^6(p-1)$ itself, since this would imply $G_p\ltimes Aut(G_p)$ is cyclic.
\end{proof}

It follows that for an arbitrary pair $(g,\phi)$ in $G_p\ltimes Aut(G_p)$, in order to compute the smallest $n$ for which $s_{g,\phi}(n)=1$, and therefore the group $\ZZ_n$, one has to compute $s_{g,\phi}(i)$ for at most 12 values of $i$. Moreover, by square-and-multiply each such computation requires $\mathcal{O}(\log p)$ applications of the group operation in the semidirect product group. In other words, we can compute a complete description of $\ZZ_n$ efficiently.

\paragraph{SDLP.} By Theorem~\ref{thm:sdlp-comp} and Theorem~\ref{thm:compute-n} we know SDLP in $G_p\ltimes Aut(G_p)$ has time complexity at most $2^{\mathcal{O}(\sqrt{\log poly(p)}}=2^{\mathcal{O}(\sqrt{\log p})}$. Taking the security parameter to be the length of an input, we can represent a pair $(g,\phi)\in G_p\ltimes Aut(G_p)$ with a bitstring of length $\mathcal{O}(\log p^2)=\mathcal{O}(\log p)$. Asymptotically, then, with $k$ as the security parameter we estimate the time complexity of the main quantum attack on SDLP as $2^{\mathcal{O}(\sqrt{k})}$. On the other hand, in order to derive a concrete estimate for specific security parameters - say, those required by NIST - one would have to check the associated constants much more carefully. Although this is outside the scope of this paper, we refer the reader to \cite[Section~7.2~{`Subexponential vs Practical'}]{castryck2018csidh} for an idea of type of spirited research carried out in pursuit of a satisfactory resolution to deriving concrete security estimates - one should note, however, that this exposition deals with specific artefacts of the isogeny framework.

\paragraph{The Dimension Attack.} Supposing an efficient method of converting an SCDH solver to an SDLP solver can be found, one solves SDLP efficiently provided one can efficiently embed $G_p$ in an algebra over a field. However, as argued in \cite{kahrobaei2016using}, the following result of Janusz \cite{janusz1971faithful} limits the effectiveness of such an approach: the smallest dimension of an algebra over a field in which a $p$-group with an element of order $p^n$ can be embedded is $1+p^{n-1}$. In our case, certainly $G_p$ has an element of order $p^2$, and so since the attack relies on Gaussian elimination we expect the dimension attack for $G_p$ to have complexity polynomial in $(p+1)^3=\bigO{p^3}$. Since the $G_p$ elements can be represented by a bitstring of order $4\log p^2=8\log p$, with $k$ the security parameter the dimension attack runs in time $\bigO{2^{3k/8}}$.

\paragraph{The Telescoping Attack.} In general, the explicit method of deducing $s_{g,\phi}(x+y)$ from $s_{g,\phi}(y)$ and $\phi^x(g)$ relies on the group $G$ being the abelian group of a matrix alegbra over a field under addition. In particular, an extension outside of this linear context is not known - we would expect, however, that such an extension would rely on equation solving techniques available only in an algebra over a field, rather than over a ring, and therefore that arguments on the efficiency of a representation discussed above would also apply.

\paragraph{Efficiency.} Multiplication in $G_p$ consists of $8$ multiplication operations and $4$ addition operations in $\ZZ_{p^2}$, for a total of $\mathcal{O}(8\log p^2)=\mathcal{O}(\log p)$ operations. Assuming that applying an automorphism $\phi$ has about the same complexity as multiplication\footnote{This is indeed the case if the automorphism is inner.}. It follows by standard square-and-multiply techniques that calculating $s_{g,\phi}$ and evaluating the group action is very roughly of complexity $\mathcal{O}((\log p)^2)$.

The signatures are also rather short, consisting of $N$ elements of $\mathcal{X}_{g,\phi}$ and $N$ elements of $\ZZ_n$. Since $\mathcal{X}_{g,\phi}\subset G_p$ we can represent $\mathcal{X}_{g,\phi}$ elements as bitstrings of length $4\log(p^2)=8\log p$; and since $n=p^i(p-1)^j$ for some $1\leq i \leq 5$ and $0\leq j\leq 1$, $\ZZ_n$ elements can be represented by bitstrings of length $\log p^i(p-1)^j$. It follows that we get signatures of length 
\[N((8+i)\log p+j\log(p-1))\]

\section{Conclusion}
We have given a constructive proof that a few elementary definitions give rise to a free, transitive group action; such a group action naturally gives rise to an identification scheme and a signature scheme. Moreover, well-known tools allow us to phrase the security of this signature scheme in terms of the semidirect discrete logarithm problem, which is itself a special case of Couveignes' Vectorisation Problem. 

Our main contributions are as follows: firstly, the generality of the construction gives an unusually diverse family of signature schemes - indeed, a signature scheme of the {\tt SPDH-Sign} type is defined for each finite group. Much further study on the relative merits of different choices of finite non-abelian group in different use cases is required to fully realise the potential of this diversity.

Second, our Theorem~\ref{thm:n-div} essentially gives us information about how to compute the group in our group action. In Theorem~\ref{thm:compute-n}, we saw one particular case where the result was enough to completely describe how to efficiently compute the group, thereby yielding an example of a group-action based key exchange in which efficient sampling is possible from the whole group, without appealing to techniques inducing additional overhead, most notably the `Fiat-Shamir with aborts' technique of Lyubashevsky.

The paper notably does not address concrete security estimates or recommend parameter sizes for a signature scheme. In order to do so we would need to carefully check the constants in the asymptotic security estimates - we consider the scale of this task, along with that of providing an implementation of the scheme, as sufficient to merit a separate paper. 

At a late stage of the preparation of this manuscript the authors were made aware of work in \cite{duman2023generic} discussing the security of group action-induced computational problems, particularly in a quantum sense. The arguments therein should be addressed when discussing the difficulty of SDLP in subsequent work.

\section*{Acknowledgements}
We would like to thank the anonymous reviewers who provided useful feedback on this manuscript.



\printbibliography

\end{document}